\documentclass[runningheads,a4paper]{llncs}

\usepackage{amsmath}
\usepackage{amssymb}
\usepackage{graphicx}

\usepackage{url}
\usepackage{slashbox}
\urldef{\mailz}\path|daze@acm.org|  
\newcommand{\keywords}[1]{\par\addvspace\baselineskip
\noindent\keywordname\enspace\ignorespaces#1}

\begin{document}
\mainmatter  

\hyphenation{multi-di-men-sion-al}

\title{$k$-neighborhood for Cellular Automata}
\titlerunning{$k$-neighborhood for Cellular Automata}

\author{Dmitry Zaitsev\thanks{Thanks to Mike Stannett for his help in improving the readability of this paper}}

\authorrunning{Zaitsev D.A.}

\institute{Faculty of Engineering \\
Vistula University \\ Stoklosy 3, 02-787 Warszawa, Poland\\
\mailz}

\toctitle{Lecture Notes in Computer Science}
\tocauthor{D. Zaitsev}
\maketitle

\begin{abstract}
A neighborhood for $d$-dimensional cellular automata is introduced that spans the range from von Neumann's to Moore's neighborhood using a parameter which represents the dimension of hypercubes connecting neighboring cells. The neighborhood is extended to include a concept of radius. The number of neighbors is calculated. For diamond-shaped neighborhoods, a sequence is obtained whose partial sums equal Delannoy numbers.
\end{abstract}

\section{Introduction}

Cellular automata \cite{Kari} are defined over certain lattices where cells change their states in either a synchronous or an asynchronous way depending on a set of local rules which specify a mapping from the state of each cell's neighborhood at the current step into a cell state at the next step. In the present paper we focus on the specification and enumeration of cell neighbors.

In cellular automata theory \cite{Kari}, two kinds of cell neighborhood, von Neumann's and Moore's, are usually considered for two dimensional space. These are then generalized to multidimensional space, and extended for radius greater than one.  

However, for multidimensional space, von Neumann's neighborhood can generate too sparse a topology, while Moore's is too dense. Our idea consists in introducing an adjustable parameter which allows us to span between these two neighborhoods. The parameter represents the dimension of hypercubes connecting neighboring cells. ``The On-Line Encyclopedia of Integer Sequences'' \cite{OEIS} has recently approved two new sequences (OEIS A265014) and (OEIS A266213) studied in the present paper and implemented with software \cite{Z15hmn}.

\section{Basic Notions}

Let us consider an infinite integer lattice of dimension $d$. Nodes of this lattice have coordinates $\vec{i}=(i_1,i_2,...,i_d)$, $i_j{\in}\mathbb{Z}$, $1{\leq}j{\leq}d$. According to the terminology of cellular automata \cite{Kari}, we call these nodes \textit{cells}, denoted  $c_{\vec{i}}$, and consider each of them to be a unit-size $d$-hypercube with its center situated with coordinates $\vec{i}$. To study neighborhoods of a cell systematically, we recall the definitions of the following distances in multidimensional space \cite{Suther}:

\begin{itemize}
\item
Minkowsky distance \cite{Minkowsky}: 
    $L^{p}(\vec{i^\prime},\vec{i})
    =  (\sum\limits_{j}\vert\vec{i'}-\vec{i}\vert)^{p})^{1/p}$.
\item
Manhattan distance \cite{Suther}: 
    $L^{1}(\vec{i'},\vec{i})
		=  \sum\limits_{j}(\vert\vec{i'}-\vec{i}\vert)$.
\item
Chebyshev distance \cite{Chebyshev}: 
    $L^{\infty}(\vec{i^\prime},\vec{i})
		=  \max\limits_{j}(\vert\vec{i'}-\vec{i}\vert)$.
\end{itemize}

Using these definitions of the above distances, we can characterize Moore's neighborhood \cite{Moore} of a cell $c_{\vec{i}}$ as the set of cells which are situated at Chebyshev distance 1, and von Neumann's neighborhood \cite{Neumann} as the set of cells which are situated at Manhattan distance 1, from $c_{\vec{i}}$. For the 2-dimensional case, two neighborhoods are illustrated in Figure~\ref{fig-nbh}. 
					
\begin{figure}
\center
\includegraphics[width=0.5\textwidth]{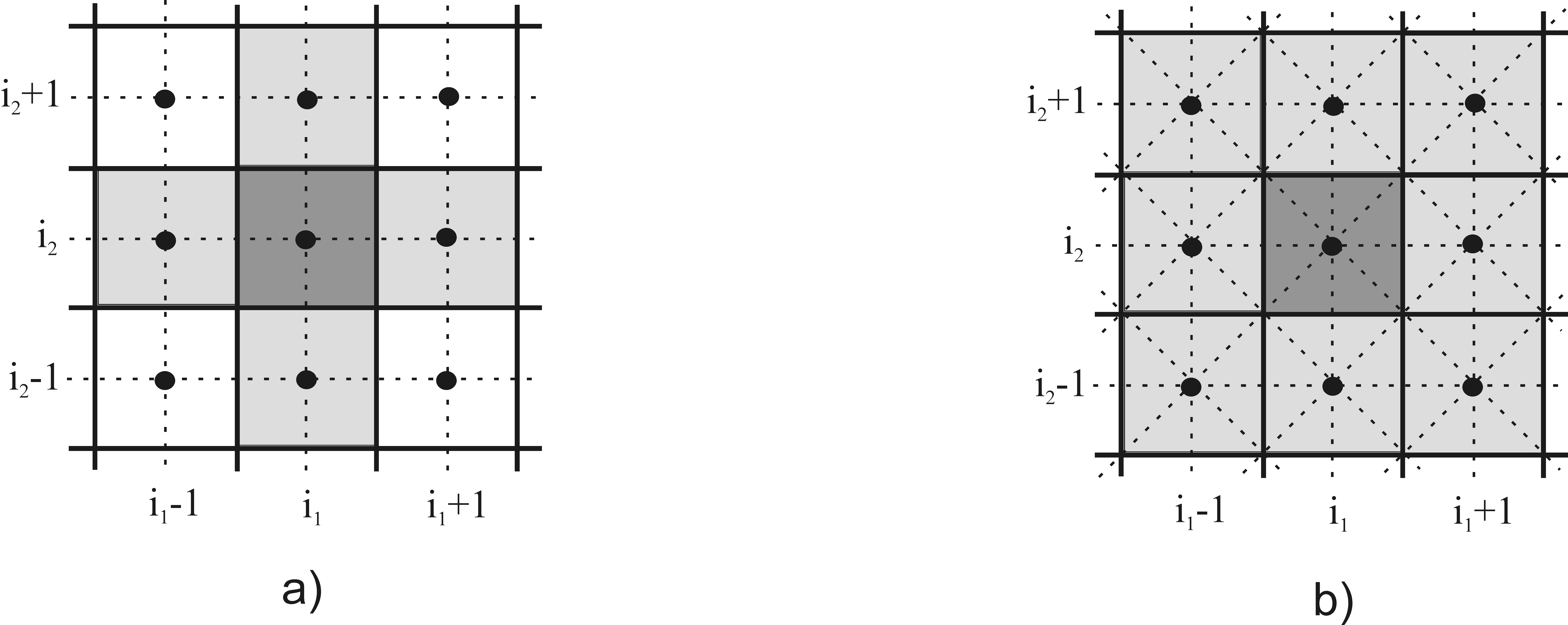}
\caption{Classical neighborhoods (2-dimensional case): a) von Neumann's neighborhood; b) Moore's neighborhood.}
\label{fig-nbh}
\end{figure}															

Note that the facets of a finite $d$-dimensional hypercube include $2d$ facets which are $(d-1)$-dimen\-sional hypercubes, each of them includes $2(d-1)$ facets which are $(d-2)$-dimensional hypercubes and so on; finally, there are $2^d$ $0$-dimensional hypercubes (i.e. vertices).  In von Neumann's neighborhood, cells are connected only via facets which are $(d-1)$-dimensional hypercubes while in Moore's neighborhood, cells are connected via bounds which are $(d-j)$-dimensional hypercubes, $1{\leq}j{\leq}d$.

\section{$k$-neighborhood}

Let us consider an infinite $d$-dimensional lattice. In von Neumann's neighborhood, the neighboring cells are situated at a Manhattan distance of 1. The number of neighbors is $2d$ and given by the sequence (OEIS A005843). This is clear if we consider the cell coordinate difference $\Delta\vec{i}=\vec{i'}-\vec{i}$, since this difference vector can contain only one nonzero element (belonging to the set  $\lbrace{-1,1}\rbrace$, which consists of two elements), and any one of the $d$ coordinates can be chosen to take this value. 

In Moore's neighborhood, the neighboring cells are situated at Chebyshev distance of 1. The number of neighbors is calculated as $(3^d-1)$ and given by the sequence (OEIS A024023). Indeed, when we consider the cell coordinate difference $\Delta\vec{i}$, its elements give precisely the set $\lbrace{-1,0,1}\rbrace^d$ (with $3^d$ elements), except that the vector having all coordinates equal to $0$ is excluded. 

\textbf{Definition 1.} A \textit{sharp $k$-neighborhood} is a set of cells having difference of either $-$1 or 1 in exactly $k$ coordinates with respect to the current cell. 

Thus, we consider neighbors connected via $(d-k)$-cube bounds of the unit-size $d$-hypercube which represents a cell. In other words, neighbors are situated at a Chebyshev distance of 1 restricted by Manhattan distance equal to $k$. We denote this neighborhood by $S(d,k)$. Noticing that only $k$ coordinates of the difference vector $\Delta\vec{i}$ (which can be chosen in $C_d^k$ ways) are nonzero, and must belong to the set $\lbrace{-1,1}\rbrace$ consisting of two elements, gives us the following formula 
\begin{equation}\label{eq-ns}
\hat{K}(d,k)=\vert{S(d,k)}\vert=2^k{C_d^k}
\end{equation}
represented by sequence (OEIS A013609). Note that, $S(d,1)$ coincides with von Neumann's neighborhood as far as $C_d^1=d$ and a union over $k$ gives us Moore's neighborhood with $\sum_{j=1}^{d} {\hat{K}(d,j)}=3^d-1$.  The diagonal numbers equal $\hat{K}(d,d)=2^d$. For instance, in the 3-dimensional case, we have 6 sharp 1-neighbors connected via 2-cube bounds (facets or squares), 12 sharp 2-neighbors connected via 1-cube bounds (sides), and 8 sharp 3-neighbors connected via 0-cube bounds (vertices).

\textbf{Definition 2.} A \textit{$k$-neighborhood} is a set of cells having difference of either $-$1 or 1 in $j$ coordinates, $1 \leq j \leq k$, with respect to the current cell.

It directly follows from the definition that the number of neighbors in a $k$-neighborhood can be calculated as:
\begin{equation}\label{eq-ng}
K(d,k)=\vert{G(d,k)}\vert=\sum\limits_{j=1}^{k} \hat{K}(d,k)=\sum\limits_{j=1}^{k} 2^{j}C_d^j
\end{equation}  
represented by a new sequence (OEIS A265014). 
Thus, a $k$-neighborhood connects cells via cubes of dimension $(d-j)$, where $1\leq{j}\leq{k}\leq{d}$. In other words, neighboring cells are situated at Chebyshev distance 1, restricted by Manhattan distance less than or equal to $k$. Of the various $k$-neighborhoods, we obtain as particular cases: von Neumann's neighborhood for $k=1$ and Moore's neighborhood for $k=d$. 

Now we are interested in efficient computation of $K(d,k)$. We know that $K(d,1)=2d$ and $K(d,d)=3^d-1$. For $\hat{K}(d,k)$, the following recurrent expression is known (OEIS A013609): 
\begin{equation}\label{eq-ns-rc}
\hat{K}(d,k)=2\hat{K}(d-1,k-1)+\hat{K}(d-1,k)
\end{equation}
combined with $\hat{K}(d,0)=1$ in the original sequence (or $\hat{K}(d,1)=2d$ in our case starting from 1). Taking into consideration the fact that $K(d,k)$ represents partial sums of $\hat{K}(d,k)$ on rows, we write: 
\[ K(d,k)=K(d,k-1)+\hat{K}(d,k), \qquad K(d,1)=\hat{K}(d,1). \]

Thus we can use a combined scheme, sequentially computing $K(d,k)$ after $\hat{K}(d,k)$ for each combination $(d,k)$. To obtain a completely separated scheme, we use (\ref{eq-ns-rc}) and write
\[ K(d,k)=K(d,k-1)+2\hat{K}(d-1,k-1)+\hat{K}(d-1,k). \]

Expressing $\hat{K}(d,k)$ via $K(d,k)$ in the following way
\[ \hat{K}(d,k)=K(d,k)-K(d,k-1) \]
yields 

\begin{equation}\label{eq-ng}
K(d,k)=K(d,k-1)+K(d-1,k)+K(d-1,k-1)-2K(d-1,k-2)).
\end{equation}

\section{$k$-neighborhoods of Radius Greater than One}

The offsets of separate coordinates to neighboring cells are usually equal to either $-$1 or 1, which gives us cells connected with the current cell via some common $(d-k)$-cube. Such a set of direct neighbors defines the mesh (lattice) of connections studied in the previous sections. But sometimes more distant cells, which are separated from the current cell by a hypercube of direct neighbors, influence its behavior. Such an influence is usually specified using a concept of radius \cite{Kari} which extends the standard notion of neighborhood, considered as a neighborhood of radius 1. Sometimes it is of some use to distinguish sharp from usual neighborhoods of definite radius $r$ by applying the equality $L(\vec i, \vec i')=r$, or the inequality $L(\vec i, \vec i')\leq r$, respectively, in the same way as for the $k$-neighborhood discussed above. A sharp neighborhood corresponds to the surface area of the corresponding figure while the usual neighborhood corresponds to its volume. Some difficulty concerns border cells because they are included twice in the usual formula for a hypercube's surface area, $2d(2r+1)^{d-1}$. Usually, the figure of a cell neighborhood defined by some metric is convex. Following the principles outlined in the previous section we will decrement the value excluding from it the current cell. We denote the number of neighbors by capital $R$ in calculations which use parameter $r$ to specify the radius of the neighborhood. Note that the number of neighbors in a standard neighborhood represents the sum of its sharp neighborhoods:
\begin{equation}
R(r)=\sum\limits_{l=1}^{r} \hat{R}(l).\label{eq-nssum}
\end{equation}

Extending Moore's neighborhood for radius $r\geq 1$ to a set of cells situated at Chebyshev distance $r$ gives us a $d$-hypercube of size $2r+1$ with the current cell in its center (which is excluded). Then the total number of neighbors is calculated as 
\begin{eqnarray}
R_{\mathit{Moore}}(d,r)=(2r+1)^d-1,\label{eq-nmrs}\\
\hat{R}_{\mathit{Moore}}(d,r)=\sum\limits_{m=1}^{d} {C_d^m}{2^m}{(2r-1)}^{d-m}\label{eq-nmr}.
\end{eqnarray}
Among the various kinds of extended neighborhood, Moore's is the most spacious. Sometimes it is convenient to think of other neighborhoods as its restrictions. 

Extending von Neumann's neighborhood for radius $r\geq 1$ to a set of cells situated at Manhattan distance $r$ gives us a diamond-shaped neighborhood, illustrated for the 2-dimensional case in Figure~\ref{fig-rdm}. Let us calculate the number of neighbors in this neighborhood. 

\begin{figure}
\center
\includegraphics[width=0.6\textwidth]{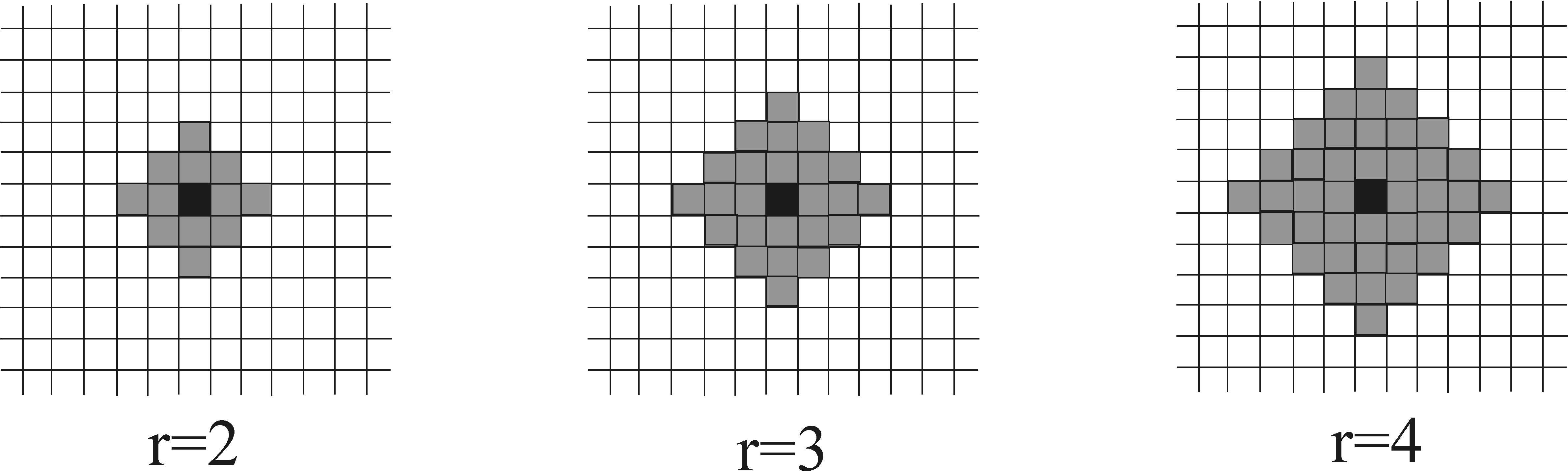}
\caption{Diamond-shaped neighborhood of radius $r$ as a set of cells situated at Manhattan distance $r$.}
\label{fig-rdm}
\end{figure}

\begin{theorem}\label{th-der}
The number of cells situated at sharp radius $r$ in a diamond-shaped neighborhood of a $d$-dimensional lattice is calculated as
$$\hat{R}_{\mathit{diamond}}(d,r)=\sum\limits_{k=1}^{\min(d,r)} C_{r-1}^{k-1} C_d^k 2^k.$$
\end{theorem}

\begin{proof}
We organize a combinatorial choice in the following way. We consider sequentially variants having exactly $k$ nonzero coordinates. At first, we choose $k$ coordinates from the total set of $d$ coordinates, which can be done in $C_d^k$ ways. Next, we partition $r$ into a sum of exactly $k$ nonzero numbers, which can be done in $C_{r-1}^{k-1}$ ways; the choice is organized as a distribution of $r$ balls between $k$ boxes \cite{Stanley}. Finally, because the Manhattan distance uses absolute values, we can choose variants of sign distribution over the $k$ nonzero values in $2^k$ ways. Multiplying the various magnitudes and summing up the required function on the number of nonzero elements $k$, we obtain the required expression.  
\end{proof}

\begin{theorem}
The number of cells situated at radius $r$ in a diamond-shaped neighborhood for a $d$-dimensional lattice is calculated as
$$R_{\mathit{diamond}}(d,r)=\sum\limits_{k=1}^{\min(d,r)} C_{r}^{k} C_d^k 2^k;$$
the numbers $D(d,r)=R_{\mathit{diamond}}(d,r)+1$, which are obtained starting the sum from zero, are known as Delannoy numbers \cite{Banderier}.
\end{theorem}

\begin{proof}
We organize a combinatorial choice in the way similar to that in the proof of Theorem~\ref{th-der}, but to cover all variants which contain $k$ nonzero components whose sum is equal to or less than $r$, we introduce a dummy box and a dummy ball which are not taken into consideration subsequently. Next, we distribute $(r+1)$ balls between $(k+1)$ boxes. Of these, $k$ boxes correspond to the vector $\vec{i}$ and the final one is a dummy used to allow distributions of all $x\leq r$ balls between $k$ boxes. The dummy box contains a surplus, the other boxes contain more than 1 ball, with the total being less than or equal to $r$ balls. Thus at this stage we obtain $C_r^k$ variants, which yields the required expression (known as Delannoy numbers \cite{Banderier}).
\end{proof}

\textbf{Corollary}. The following equality holds:
$$D(d,r)=\sum\limits_{k=1}^{\min(d,r)} C_{r}^{k} C_d^k 2^k = \sum\limits_{l=1}^{r} \sum\limits_{k=1}^{\min(d,l)} C_{l-1}^{k-1} C_d^k 2^k$$

To prove the Corollary, we use (\ref{eq-nssum}). The expression gives us one more variant for computing Delannoy numbers, showing that they can be decomposed into partial sums of a new useful sequence $\hat{R}_{\mathit{diamond}}(d,r)$ represented with (OEIS A266213). Note that, \cite{Bruck} includes a table showing the size of a diamond-shaped neighborhood of radius $r$ comparing to Moore's neighborhood of radius $r$, but without mentioning Delannoy numbers and without proof. 

Delannoy numbers \cite{Banderier},  sequence (OEIS A008288), are represented as a triangle using its detour on anti-diagonals. Note that they were initially introduced to describe the number of paths in a 2-dimensional rectangular $d\times r$ lattice from its left-bottom corner with coordinates $(0,0)$ to its right-upper corner with coordinates $(d,r)$ using only the three following steps: $(1,0)$, $(0,1)$, $(1,1)$. Delannoy numbers are also known as the tribonacci triangle because of the following efficient recursive scheme for their computation using three previously computed members:

$$D(d,r)=1,~if~d=1~or~r=1,$$
$$D(d,r)= D(d-1,r)+D(d-1,r-1)+D(d,r-1).$$

Note that the same scheme is valid for $\hat{R}_{\mathit{diamond}}(d,r)$ with the following initial conditions: 
$$ \hat{R}_{\mathit{diamond}}(d,0)=1, d\geq 0; \qquad \hat{R}_{\mathit{diamond}}(0,r)=0, r>0 $$

There are also other known ways for extending von Neumann's neighborhood using radius $r > 1$; for example the scheme outlined in \cite{Dutta}. This only allows differences with absolute value equal to $r$ in a single coordinate and is illustrated in Figure~\ref{fig-rnm} for the 2-dimensional case. We will call this extension a narrow von Neumann's neighborhood of radius $r$. It suits our intention to extend the $k$-neighborhood studied in previous section to extended neighborhoods of radius $r>1$. 

\textbf{Definition 3}. A \textit{$k$-neighborhood of radius $r$} is a set of cells having difference $\vert\Delta\vert\leq r$ with respect to the current cell in $j\leq k$ coordinates. When using equalities instead of inequalities we say that the neighborhoods are sharp on $r$ or $k$ (or both), respectively. 

Note that the parameter $k$ allows us to span neighborhoods of radius $r$ from narrow von Neumann's ($k=1$) to Moore's ($k=d$). However, using neighborhoods sharp on $k$ seems unmotivated because it contains gaps.

\begin{figure}
\center
\includegraphics[width=0.6\textwidth]{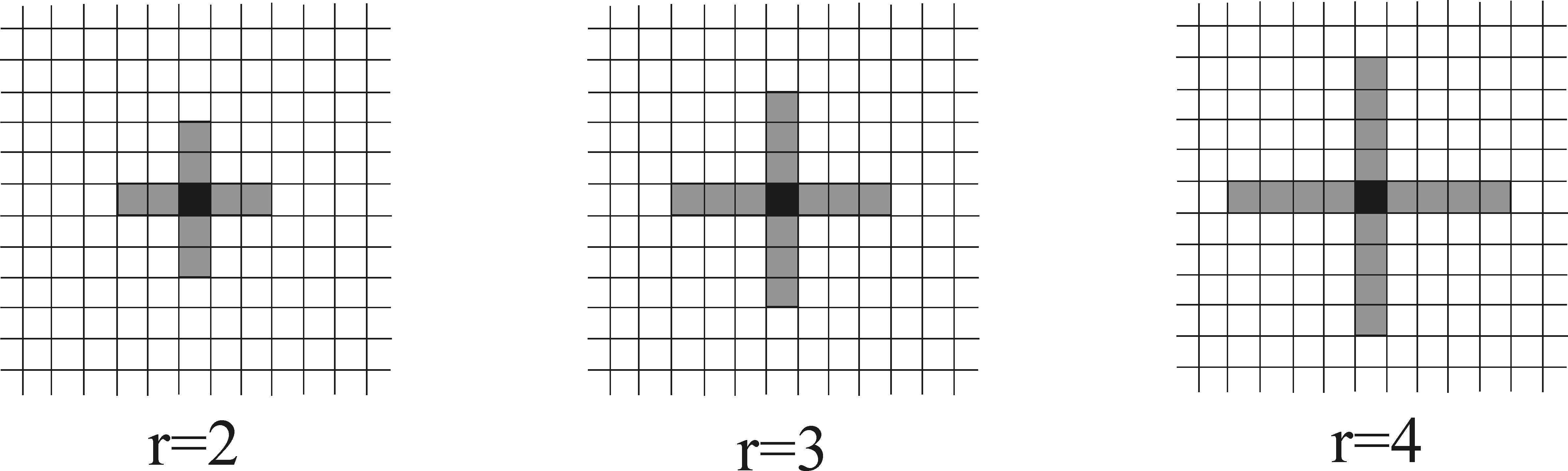}
\caption{Narrow von Neumann's neighborhood of radius $r$.}
\label{fig-rnm}
\end{figure}

A $k$-neighborhood of radius $r$ is the union of all $(d-k)$-cubes of size $(2r+1)$ which have their center in the current cell. However, when calculating the number of neighbors as a sum of cells in $(d-k)$-cubes for all possible combinations of $k$-from-$d$ coordinates, we should exclude intersections of cubes which are summed up multiple times. 

\begin{theorem}
The number of neighbors in a $k$-neighborhood of radius $r$ is calculated using the following expression:
\begin{eqnarray}
R(d,k,r) &=& \sum\limits_{j=0}^{k} C_{d}^{j} (2r)^j \label{eq-dkr1}
\end{eqnarray}
\end{theorem}

\begin{proof}
To prove (\ref{eq-dkr1}) we organize a combinatorial choice on the number $j$ of nonzero components of the coordinates offset vector $\Delta{\vec i}$. For each component, there are $2r$ variants since zero is excluded; thus, we obtain $(2r)^j$ combinations of nonzero coordinates. Next, $j$ nonzero coordinates are chosen from $d$ coordinates in $C_d^j$ ways which give us the corresponding multiplier. 
\end{proof}

\section{Conclusions}

For $d$-dimensional cellular automata, we have introduced $k$-neighborhood based on a parameter $k$ that corresponds to connections of neighboring cells via common $(d-k)$-cubes. We distinguish sharp and usual neighborhoods, according to whether equality of inequality on the parameter $k$ is considered. Usual $k$-neighborhoods span from von Neumann's neighborhood ($k=1$) to Moore's neighborhood ($k=d$). This concept is useful for adjusting the density of neighbors, which increases with increasing $k$. The $k$-neighborhood was extended using the concept of radius $r$ greater than one. We studied both von Neumann's neighborhoods of radius $r$: the narrow corresponding to $k=1$ and the diamond-shaped corresponding to a set of cells situated at Manhattan distance $r$. Calculating the number of neighbors for the sharp case provided a new expression which partial sums equal Delannoy numbers. Two new sequences (OEIS A265014) and (OEIS A266213) have been approved by OEIS \cite{OEIS} and implemented in software \cite{Z15hmn}.


\end{document}